\documentclass[11pt]{article}%
\usepackage{amsmath}
\usepackage{amsfonts}
\usepackage{amssymb}
\usepackage{graphicx}%
\newtheorem{theorem}{Theorem}

\newtheorem{corollary}[theorem]{Corollary}

\newtheorem{definition}[theorem]{Definition}
\newtheorem{example}[theorem]{Example}

\newtheorem{lemma}[theorem]{Lemma}

\newtheorem{remark}[theorem]{Remark}

\newenvironment{proof}[1][Proof]{\noindent\textbf{#1.} }{\ \rule{0.5em}{0.5em}}
\begin{document}
\centerline{{\large \textbf{Factorization of Difference Equations by Semiconjugacy} }}
\vspace{1ex}
\centerline{{\large \textbf{with Application to Non-autonomous Linear Equations} }}

\vspace{4ex}

\centerline{H. SEDAGHAT*} \footnotetext{\noindent*Department of Mathematics,
Virginia Commonwealth University, Richmond, Virginia 23284-2014
\par
USA, Email: hsedagha@vcu.edu}

\vspace{3ex}

\begin{quote}
\noindent{\small \textbf{Abstract.} The existence of a semiconjugate relation
permits the transformation of a higher order difference equation on a group
into an equivalent triangular system of two difference equations of lower
orders. Introducing time-dependent form symmetries in this paper enables us to
identify the semiconjugate property in a larger set of non-autonomous
difference equations than previously considered. We show that there is a
substantial class of equations having this feature that includes the general
(non-autonomous, non-homogeneous) linear equation with variable coefficients
in an arbitrary algebraic field.}
\end{quote}

\section{Introduction}

Difference equations of order greater than one that are of the following type%
\begin{equation}
x_{n+1}=f_{n}(x_{n},x_{n-1},\ldots,x_{n-k}) \label{dek}%
\end{equation}
determine the forward evolution of a variable $x_{n}$ in discrete time since
the time index or the independent variable $n$ is integer-valued with $n\geq0$.

In previous studies of semiconjugate factorizations of difference equations of
type (\ref{dek}), e.g., \cite{hd1}, \cite{hsinvcrt}, \cite{hsarx} or
\cite{hsijpam}, the form symmetry linking the higher dimensional unfolding map
of the original equation to that of the lower dimensional factor was assumed
to be independent of $n$. While this assumption did not substantially curtail
the applicability of the method, it did rule out certain non-autonomous
equations. For example, the method worked for non-homogeneous linear equations
with constant coefficients but did not apply to linear equations with variable coefficients.

The main goal of this article is to extend the aforementioned factorization
method to allow \textit{time-dependent} form symmetries where the form
symmetry may depend explicitly on the independent variable $n$. This extension
is significant as it covers all non-autonomous equations of type (\ref{dek}).
In particular, the extended method may be applied to general (non-autonomous,
non-homogeneous) linear equations over arbitrary algebraic fields to show that
such equations admit semiconjugate factorizations via eigensequences (i.e.,
the solutions of an associated discrete Riccati difference equation of lower
order). For ease of reference we state some of the basic concepts and notation
here; additional background material for this article is available in
\cite{hsarx}.

As usual, the number $k$ in (\ref{dek}) is a fixed positive integer and $k+1$
represents the \textit{order} of the difference equation (\ref{dek}). The
underlying space of variables $x_{n}$ is a group $G$ and $f_{n}:G^{k+1}%
\rightarrow G$ is a given function for each $n\geq1$. If $f_{n}=f$ does not
explicitly depend on $n$ then (\ref{dek}) is said to be \textit{autonomous};
it is \textit{non-autonomous }otherwise. A (forward) \textit{solution} of
Eq.(\ref{dek}) is a sequence $\{x_{n}\}_{n=-k}^{\infty}$ that is recursively
generated by (\ref{dek}) from a set of $k+1$ initial values $x_{0}%
,x_{-1},\ldots,x_{-k}\in G.$ Forward solutions have traditionally been of
greater interest in discrete models that are based on Eq.(\ref{dek}) although
other types of solutions (e.g., those having domain $\mathbb{Z}$, the set of
all integers) can also be readily defined.

Each $f_{n}$ is \textquotedblleft unfolded\textquotedblright\ by the
associated vector map $F_{n}:G^{k+1}\rightarrow G^{k+1}$ that are defined as%
\begin{equation}
F_{n}(u_{0},\ldots,u_{k})=[f_{n}(u_{0},\ldots,u_{k}),u_{0},\ldots
,u_{k-1}],\quad u_{j}\in G\text{ for }j=0,1,\ldots,k. \label{unfold}%
\end{equation}

The \textit{unfoldings} $F_{n}$ determine the equation
\[
(y_{0,n+1},y_{1,n+1},\ldots,y_{k,n+1})=F_{n}(y_{0,n},y_{1,n},\ldots,y_{k,n})
\]
in $G^{k+1}$. Each vector $(y_{0,n+1},\ldots,y_{k,n+1})$ represents a
\textit{state }of the system, or of Eq.(\ref{dek}); $G^{k+1}$ is the
\textit{state space}, in analogy to the phase space in differential equations.

\section{Semiconjugate relation and factorization}

Let $F_{n}$ be the unfolding on $G^{k+1}$ of $f_{n}$ for each $n.$ Then
(\ref{dek}) is equivalent to%
\begin{equation}
X_{n+1}=F_{n}(X_{n}),\quad X_{n}=(x_{n},\ldots,x_{n-k}). \label{recgenmp}%
\end{equation}

We are interested in deriving a lower dimensional equation%
\begin{equation}
Y_{n+1}=\Phi_{n}(Y_{n}),\quad Y_{n}=(y_{n},\ldots,y_{n-m+1}),\ m\leq k
\label{rgmy}%
\end{equation}
for (\ref{recgenmp}). If there exists a sequence of maps $H_{n}:G^{k+1}%
\rightarrow G^{m}$ such that for every solution $\{X_{n}\}$ of (\ref{recgenmp}%
)
\begin{equation}
Y_{n}=H_{n}(X_{n}),\quad n=0,1,2,\ldots\label{yhx}%
\end{equation}
is a solution of (\ref{rgmy}) then
\[
\Phi_{n}(H_{n}(X_{n}))=\Phi_{n}(Y_{n})=Y_{n+1}=H_{n+1}(X_{n+1})=H_{n+1}%
(F_{n}(X_{n})).
\]

Therefore, (\ref{yhx}) is satisfied for all solutions of (\ref{recgenmp}) and
(\ref{rgmy}) if and only if the sequence $\{H_{n}\}$ of maps satisfies the
following equality for all $n$%
\begin{equation}
H_{n+1}\circ F_{n}=\Phi_{n}\circ H_{n}. \label{tdscr}%
\end{equation}

If the mappings $H_{n}$ are independent of $n,$ i.e., $H_{n}=H$ for all $n$
then Eq.(\ref{tdscr}) reduces to
\begin{equation}
H\circ F_{n}=\Phi_{n}\circ H \label{scr}%
\end{equation}
namely, the time-independent semiconjugate relation\ as defined in prior
studies. We can now give the following more general defintion.

\begin{definition}
\label{tdfsdef}\textit{Let }$k\geq1$\textit{, }$1\leq m\leq k.$ If there is a
sequence of surjective maps $H_{n}:G^{k+1}\rightarrow G^{m}$ such that
Eq.(\ref{tdscr}) is satisfied for a given pair of function sequences
$\{F_{n}\}$ and $\{\Phi_{n}\}$ then we say that $F_{n}$ is
\textbf{semiconjugate} to $\Phi_{n}$ for each $n$ and refer to the sequence
$\{H_{n}\}$ as a \textbf{(time-dependent) form symmetry} of Eq.(\ref{recgenmp}%
) or equivalently, of Eq.(\ref{dek}). Since $m<k+1,$ the form symmetry
$\{H_{n}\}$ is \textbf{order-reducing}.
\end{definition}

Technically, a time-dependent form symmetry can also be defined as a
\textit{single map}
\[
H:\mathbb{N}\times G^{k+1}\rightarrow G^{m},\quad H(n;u_{0},\ldots
,u_{k})=H_{n}(u_{0},\ldots,u_{k})
\]

We choose the sequence definition due to its more intuitive content.

The following result extends its time-independent analog in \cite{hsarx} and
makes precise the concept of semiconjugate factorization for the recursive
difference equation (\ref{dek}).

\begin{lemma}
\label{tdfsthm}\textit{Let }$k\geq1$\textit{, }$1\leq m\leq k$\textit{, let
}$h_{n}:G^{k-m+1}\rightarrow G$\textit{\ for }$n\geq-m+1$\textit{\ be a
sequence of functions on a given non-trivial group }$G$ and define the
functions $H_{n}:G^{k+1}\rightarrow G^{m}$ by\textit{\ }%

\begin{equation}
H_{n}(u_{0},\ldots,u_{k})=[u_{0}\ast h_{n}(u_{1},\ldots,u_{k+1-m}%
),\ldots,u_{m-1}\ast h_{n-m+1}(u_{m},\ldots,u_{k})]. \label{tdfs}%
\end{equation}
Then the following statements are true:\textit{\ }

(a) The function $H_{n}$ defined by (\ref{tdfs}) is surjective for each fixed
$n\geq0$.

(b) \textit{If }$\{H_{n}\}$\textit{\ is an order-reducing form symmetry then
the difference equation (\ref{dek})\ is equivalent to the system of equations
}%

\begin{align}
t_{n+1}  &  =\phi_{n}(t_{n},\ldots,t_{n-m+1}),\label{tdf}\\
x_{n+1}  &  =t_{n+1}\ast h_{n+1}(x_{n},\ldots,x_{n-k+m})^{-1} \label{tdcf}%
\end{align}

\textit{whose orders }$m$\textit{\ and }$k+1-m$\textit{\ respectively, add up
to the order of (\ref{dek}).}

(c) The map $\Phi_{n}:G^{m}\rightarrow G^{m}$ in (\ref{tdscr}) is the
unfolding of Eq.(\ref{tdf}) for each $n\geq0$; i.e., each $\Phi_{n}$ is of
scalar type.
\end{lemma}

\begin{proof}
(a) Let $n$ be a fixed non-negative integer and for $j=0,\ldots,m-1$ denote
the $j$-th coordinate function of $H_{n}$ by
\begin{equation}
\eta\,_{j+1}(u_{0},\ldots,u_{k})=u_{j}\ast h_{n-j}(u_{j+1},\ldots,u_{j+k+1-m})
\label{etaj}%
\end{equation}

Now choose an arbitrary point $(v_{1},\ldots,v_{m})\in G^{m}$ and define
\begin{align}
u_{m-1}  &  =v_{m}\ast h_{n-m+1}(u_{m},u_{m+1}\ldots,u_{k})^{-1},\label{umv}\\
u_{m}  &  =u_{m+1}=\ldots u_{k}=\bar{u}\nonumber
\end{align}

where $\bar{u}$ is a fixed element of $G,$ e.g., the identity.\ Then
\begin{align*}
v_{m}  &  =u_{m-1}\ast h_{n-m+1}(\bar{u},\bar{u}\ldots,\bar{u})\\
&  =u_{m-1}\ast h_{n-m+1}(u_{m},u_{m+1}\ldots,u_{k})\\
&  =\eta_{m}(u_{0},\ldots,u_{k})\\
&  =\eta_{m}(u_{0},\ldots,u_{m-2},\underset{u_{m-1}}{\underbrace{v_{m}\ast
h_{n-m+1}(\bar{u},\bar{u}\ldots,\bar{u})^{-1}}},\bar{u}\ldots,\bar{u}).
\end{align*}
for any selection of elements $u_{0},\ldots,u_{m-2}\in G.$ Using the same
idea, define
\[
u_{m-2}=v_{m-1}\ast h_{n-m+2}(u_{m-1},\bar{u}\ldots,\bar{u})^{-1}%
\]
with $u_{m-1}$ defined by (\ref{umv}) so as to get
\begin{align*}
v_{m-1}  &  =u_{m-2}\ast h_{n-m+2}(u_{m-1},\bar{u}\ldots,\bar{u})\\
&  =u_{m-2}\ast h_{n-m+2}(u_{m-1},u_{m}\ldots,u_{k-1})\\
&  =\eta_{m-1}(u_{0},\ldots,u_{k})\\
&  =\eta_{m-1}(u_{0},\ldots,u_{m-3},\underset{u_{m-2}}{\underbrace{v_{m-1}\ast
h_{n-m+2}(u_{m-1},\bar{u}\ldots,\bar{u})^{-1}}},u_{m-1},\bar{u}\ldots,\bar{u})
\end{align*}
for any choice of $u_{0},\ldots,u_{m-3}\in G.$ Continuing in this way, by
induction we obtain elements $u_{m-1},\ldots,u_{0}\in G$ such that%

\[
v_{i}=\eta_{i}(u_{0},\ldots,u_{m-1},\bar{u}\ldots,\bar{u}),\quad
i=1,\ldots,m.
\]
Therefore, $H_{n}(u_{0},\ldots,u_{m-1},\bar{u}\ldots,\bar{u})=(v_{1}%
,\ldots,v_{m})$ and it follows that $H_{n}$ is onto $G^{m}.$

(b) To show that the SC factorization system consisting of equations
(\ref{tdf}) and (\ref{tdcf}) is equivalent to Eq.(\ref{dek}) we show that: (i)
each solution $\{x_{n}\}$ of (\ref{dek}) uniquely generates a solution of the
system (\ref{tdf}) and (\ref{tdcf}) and conversely (ii) each solution
$\{(t_{n},y_{n})\}$ of the system (\ref{tdf}) and (\ref{tdcf}) correseponds
uniquely to a solution $\{x_{n}\}$ of (\ref{dek}). To establish (i) let
$\{x_{n}\}$ be the unique solution of (\ref{dek}) corresponding to a given set
of initial values $x_{0},\ldots x_{-k}\in G.$ Define the sequence
\begin{equation}
t_{n}=x_{n}\ast h_{n}(x_{n-1},\ldots,x_{n-k+m-1}) \label{tntd}%
\end{equation}

\noindent for $n\geq-m+1.$ Then for each $n\geq0$ if $H_{n}$ is defined by
(\ref{tdfs}) it follows from the semiconjugate relation (\ref{tdscr}) that%
\begin{align*}
x_{n+1}  &  =f_{n}(x_{n},\ldots,x_{n-k})\\
&  =\phi_{n}(x_{n}\ast h_{n}(x_{n-1},\ldots,x_{n-k+m-1}),\ldots,\\
&  \qquad x_{n-m+1}\ast h_{n-m+1}(x_{n-m},\ldots,x_{n-k}))\ast\lbrack
h_{n+1}(x_{n},\ldots,x_{n-k+m})]^{-1}\\
&  =\phi_{n}(t_{n},\ldots,t_{n-m+1})\ast\lbrack h_{n+1}(x_{n},\ldots
,x_{n-k+m})]^{-1}%
\end{align*}

Therefore, $\phi_{n}(t_{n},\ldots,t_{n-m+1})=x_{n+1}\ast h_{n+1}(x_{n}%
,\ldots,x_{n-k+m})=t_{n+1}$ so that $\{t_{n}\}$ is the unique solution of the
factor equation (\ref{tdf}) with initial values%
\[
t_{-j}=x_{-j}\ast h_{-j}(x_{-j-1},\ldots,x_{-j-k+m-1}),\quad j=0,\ldots,m-1.
\]

Further, since $x_{n+1}=t_{n+1}\ast\lbrack h_{n+1}(x_{n},\ldots,x_{n-k+m}%
)]^{-1}$ for $n\geq0$ by (\ref{tntd}), $\{x_{n}\}$ is the unique solution of
the cofactor equation (\ref{tdcf}) with initial values $y_{-i}=x_{-i}$ for
$i=0,1,\ldots,k-m$ and with the values $t_{n}$ obtained above.

To establish (ii) let $\{(t_{n},y_{n})\}$ be a solution of the factor-cofactor
system with initial values
\[
t_{0},\ldots,t_{-m+1},y_{-m},\ldots y_{-k}\in G.
\]

Note that these numbers determine $y_{-m+1},\ldots,y_{0}$ through the cofactor
equation%
\begin{equation}
y_{-j}=t_{-j}\ast\lbrack h_{-j}(y_{-j-1},\ldots,y_{-j-1-k+m})]^{-1},\quad
j=0,\ldots,m-1. \label{y-ntd}%
\end{equation}

Now for $n\geq0$ we obtain%
\begin{align*}
y_{n+1}  &  =t_{n+1}\ast\lbrack h_{n+1}(y_{n},\ldots,y_{n-k+m})]^{-1}\\
&  =\phi_{n}(t_{n},\ldots,t_{n-m+1})\ast\lbrack h_{n+1}(y_{n},\ldots
,y_{n-k+m})]^{-1}\\
&  =\phi_{n}(y_{n}\ast h_{n}(y_{n-1},\ldots,y_{n-k+m-1}),\ldots,\\
&  \qquad y_{n-m+1}\ast h_{n-m+1}(y_{n-m},\ldots,y_{n-k}))\ast h_{n+1}%
(y_{n},\ldots,y_{n-k+m})^{-1}\\
&  =f_{n}(y_{n},\ldots,y_{n-k})
\end{align*}

Thus $\{y_{n}\}$ is the unique solution of Eq.(\ref{dek}) that is generated by
the initial values (\ref{y-ntd}) and $y_{-m},\ldots y_{-k}.$ This completes
the proof of (b).

(c) We show that each coordinate function $\phi_{j,n}$ is the projection into
coordinate $j-1$ for $j>1.$ From the definition of $H_{n}$ in (\ref{tdfs}) and
the semiconjugate relation (\ref{tdscr}) we infer that%
\begin{align*}
H_{n+1}(F_{n}(u_{0},\ldots,u_{k}))  &  =H_{n+1}(f_{n}(u_{0},\ldots
,u_{k}),u_{0},\ldots,u_{k-1})\\
&  =(f_{n}(u_{0},\ldots,u_{k})\ast h_{n+1}(u_{0},\ldots,u_{k-m}),\\
&  u_{0}\ast h_{n}(u_{1},\ldots,u_{k-m+1}),\ldots,\\
&  u_{m-2}\ast h_{n-m+2}(u_{m-1},\ldots,u_{k-1})).
\end{align*}

Matching the corresponding component functions in the above equality for
$j\geq2$ yields%
\begin{gather*}
\phi_{j,n}(u_{0}\ast h_{n}(u_{1},\ldots,u_{k+1-m}),\ldots,u_{m-1}\ast
h_{n-m+1}(u_{m},\ldots,u_{k}))=\\
u_{j-2}\ast h(u_{j-1},u_{j}\ldots,u_{j+k-m-1})
\end{gather*}
which shows that $\phi_{j,n}$ maps its $j$-th coordinate to its $(j-1)$-st.
Therefore, for each $n$ and every $(t_{1},\ldots,t_{m})\in H_{n}(G^{k+1})$ we
have%
\[
\Phi_{n}(t_{1},\ldots,t_{m})=[\phi_{n}(t_{1},\ldots,t_{m}),t_{1}%
,\ldots,t_{m-1}]
\]
i.e., $\Phi_{n}|_{H_{n}(G^{k+1})}$ is of scalar type. Since by Part (a)
$H_{n}(G^{k+1})=G^{m}$ for every $n,$ it follows that $\Phi_{n}$ is of scalar type.
\end{proof}

The pair of equations (\ref{tdf}) and (\ref{tdcf}) in Theorem \ref{tdfsthm} is
uncoupled in the sense that (\ref{tdf}) is independent of (\ref{tdcf}). Such a
pair forms a triangular system as defined in \cite{al} and \cite{js}. In the
next definition we use convenient and suggestive terminology to describe these equations.

\begin{definition}
Eq.(\ref{tdf}) is a \textbf{factor} of Eq.(\ref{dek}) since it is derived from
the semiconjugate factor $\Phi_{n}.$ Eq.(\ref{tdcf}) that links the factor to
the original equation is a \textbf{cofactor} of Eq.(\ref{dek}). We refer to
the system of equations (\ref{tdf}) and (\ref{tdcf}) as a
\textbf{semiconjugate (SC) factorization} of Eq.(\ref{dek}). Note that
\textit{orders }$m$\textit{\ and }$k+1-m$\textit{\ of (\ref{tdf}) and
(\ref{tdcf}) respectively, add up to the order of (\ref{dek}). We refer to the
system of equations (\ref{tdf}) and (\ref{tdcf}) as a \textbf{type-}%
(}$m,k+1-m$) \textbf{order reduction} of Eq.\textit{(\ref{dek}).}
\end{definition}

\section{Invertible-map criterion}

In \cite{hsinvcrt} and \cite{hsarx} a useful necessary and sufficient
condition is obtained by which to determine whether the difference equation
(\ref{dek}) has order-reducing form symmetries (not time-dependent). In this
section we show that the same useful idea extends to the time-dependent
context. Applications and examples are discussed in the next section.

Consider the following special case of (\ref{tdfs}) with $m=k$
\begin{equation}
H_{n}(u_{0},u_{1},\ldots,u_{k})=[u_{0}\ast h_{n}(u_{1}),u_{1}\ast
h_{n-1}(u_{2}),\ldots,u_{k-1}\ast h_{n-k+1}(u_{k})] \label{tdfsk1}%
\end{equation}
with $h_{n}:G\rightarrow G$ being a sequence of surjective self-maps of the
underlying group $G$ for $n\geq-k+1$. If (\ref{dek}) has the form symmetry
(\ref{tdfsk1}) then it admits a type-$(k,1)$ order-reduction and its SC
factorization is%
\begin{align}
t_{n+1}  &  =\phi_{n}(t_{n},\ldots,t_{n-k+1}),\label{tdf1}\\
x_{n+1}  &  =t_{n+1}\ast h_{n+1}(x_{n})^{-1}. \label{tdfc1}%
\end{align}

The initial values of the factor equation (\ref{tdf1}) are%
\[
t_{-j}=x_{-j}\ast h_{-j}(x_{-j+1}),\quad j=0,1,\ldots,k-1
\]

\begin{theorem}
\label{tdhinv}(Time-dependent invertible map criterion) Assume that
$h_{n}:G\rightarrow G$ is a sequence of bijections of $G$ for $n\geq-k+1$. For
arbitrary elements $u_{0},v_{1},\ldots,v_{k}\in G$ and every $n\geq0$ define
$\zeta_{0,n}(u_{0})\equiv u_{0}$ and for $j=1,\ldots,k,$
\begin{equation}
\zeta_{j,n}(u_{0},v_{1},\ldots,v_{j})=h_{n-j+1}^{-1}(\zeta_{j-1,n}(u_{0}%
,v_{1},\ldots,v_{j-1})^{-1}\ast v_{j}). \label{zetajn}%
\end{equation}
with the usual distinction observed between map inversion and group inversion.
Then Eq.(\ref{dek}) has the form symmetry $\{H_{n}\}$ defined by
(\ref{tdfsk1}) if and only if the quantity
\begin{equation}
f_{n}(\zeta_{0,n},\zeta_{1,n}(u_{0},v_{1}),\ldots,\zeta_{k,n}(u_{0}%
,v_{1},\ldots,v_{k}))\ast h_{n+1}(u_{0}) \label{tdhinvcrit}%
\end{equation}
is independent of $u_{0}$ for every $n\geq0$.

In this case Eq.(\ref{dek}) has a SC factorization whose factor functions in
(\ref{tdf1}) are given by%
\begin{equation}
\phi_{n}(v_{1},\ldots,v_{k})=f_{n}(\zeta_{0,n},\zeta_{1,n}(u_{0},v_{1}%
),\ldots,\zeta_{k,n}(u_{0},v_{1},\ldots,v_{k}))\ast h_{n+1}(u_{0}).
\label{tdinvf}%
\end{equation}

\end{theorem}

\begin{proof}
Assume first that (\ref{tdhinvcrit}) is independent of $u_{0}$ for all
$v_{1},\ldots,v_{k}$ so that the functions%
\begin{equation}
\phi_{n}(v_{1},\ldots,v_{k})=f_{n}(\zeta_{0,n},\zeta_{1,n},\ldots,\zeta
_{k,n})\ast h_{n+1}(u_{0}) \label{tdgn}%
\end{equation}
are well defined. Next, if $H_{n}$ is given by (\ref{tdfsk1}) then for all
$u_{0},u_{1},\ldots,u_{k}$
\[
\phi_{n}(H_{n}(u_{0},u_{1},\ldots,u_{k}))=\phi_{n}(u_{0}\ast h_{n}%
(u_{1}),u_{1}\ast h_{n-1}(u_{2}),\ldots,u_{k-1}\ast h_{n-k+1}(u_{k})).
\]

Now, by (\ref{zetajn}) for each $n$ and all $u_{0},u_{1}$
\[
\zeta_{1,n}(u_{0},u_{0}\ast h_{n}(u_{1}))=h_{n}^{-1}(u_{0}^{-1}\ast u_{0}\ast
h_{n}(u_{1}))=u_{1}.
\]

Similarly, for each $n$ and all $u_{0},u_{1},u_{2}$
\begin{align*}
\zeta_{2,n}(u_{0},u_{0}\ast h_{n}(u_{1}),u_{1}\ast h_{n-1}(u_{2}))  &
=h_{n-1}^{-1}(\zeta_{1,n}(u_{0},u_{0}\ast h_{n}(u_{1}))^{-1}\ast\\
&  \qquad\qquad u_{1}\ast h_{n-1}(u_{2}))\\
&  =u_{2}.
\end{align*}

Suppose by way of induction that%
\[
\zeta_{l,n}(u_{0}\ast h_{n}(u_{1}),\ldots,u_{l-1}\ast h_{n-l+1}(u_{k}))=u_{l}%
\]
for $1\leq l<j$. Then
\[
\zeta_{j,n}(u_{0}\ast h_{n}(u_{1}),\ldots,u_{j-1}\ast h_{n-j+1}(u_{j}%
))=h_{n-j+1}^{-1}(u_{j-1}^{-1}\ast u_{j-1}\ast h_{n-j+1}(u_{j}))=u_{j}.
\]

Thus by (\ref{tdgn})
\[
\phi_{n}(H_{n}(u_{0},u_{1},\ldots,u_{k}))=f_{n}(u_{0},\ldots,u_{k})\ast
h_{n+1}(u_{0})
\]

Now if $F_{n}$ and $\Phi_{n}$ are the unfoldings of $f_{n}$ and $\phi_{n}$
respectively, then%

\begin{align*}
H_{n+1}(F_{n}(u_{0},\ldots,u_{k}))  &  =[f_{n}(u_{0},\ldots,u_{k})\ast
h_{n+1}(u_{0}),u_{0}\ast h_{n}(u_{1}),\\
&  \qquad\ldots,u_{k-2}\ast h_{n-k+2}(u_{k-1})]\\
&  =[\phi_{n}(H_{n}(u_{0},u_{1},\ldots,u_{k})),u_{0}\ast h_{n}(u_{1}),\\
&  \qquad\ldots,u_{k-2}\ast h_{n-k+2}(u_{k-1})]\\
&  =\Phi_{n}(H_{n}(u_{0},\ldots,u_{k}))
\end{align*}
and it follows that $\{H_{n}\}$ is a semiconjugate form symmetry for
Eq.(\ref{dek}). The existence of a SC factorization with factor functions
defined by (\ref{tdinvf}) now follows from Lemma \ref{tdfsthm}.

Conversely, if $\{H_{n}\}$ as given by (\ref{tdfsk1}) is a time-dependent form
symmetry of Eq.(\ref{dek}) then the semiconjugate relation implies that for
arbitrary $u_{0},\ldots,u_{k}$ in $G$ there are functions $\phi_{n}$ such that%
\begin{equation}
f_{n}(u_{0},\ldots,u_{k})\ast h_{n+1}(u_{0})=\phi_{n}(u_{0}\ast h_{n}%
(u_{1}),\ldots,u_{k-1}\ast h_{n-k+1}(u_{k})). \label{tdgn1}%
\end{equation}

For every $u_{0},v_{1},\ldots,v_{k}$ in $G$ and with functions $\zeta_{j,n} $
as defined above, note that%
\[
\zeta_{j-1,n}(u_{0},v_{1},\ldots,v_{j-1})\ast h_{n-j+1}(\zeta_{j,n}%
(u_{0},v_{1},\ldots,v_{j}))=v_{j},\quad j=1,2,\ldots,k.
\]

Therefore, abbreviating $\zeta_{j,n}(u_{0},v_{1},\ldots,v_{j})$ by
$\zeta_{j,n}$ we have
\begin{align*}
f_{n}(\zeta_{0,n},\zeta_{1,n},\ldots,\zeta_{k,n})\ast h_{n+1}(u_{0})  &
=\phi_{n}(\zeta_{0,n}\ast h_{n}(\zeta_{1,n}),\zeta_{1,n}\ast h_{n-1}%
(\zeta_{2,n}),\\
&  \quad\ldots,\zeta_{k-1,n}\ast h_{n-k+1}(\zeta_{k,n}))\\
&  =\phi_{n}(v_{1},\ldots,v_{k})
\end{align*}
which is independent of $u_{0}.$
\end{proof}

Recall that an algebraic field $\mathcal{F=}(\mathcal{F},+,\cdot)$ is, in
particular, a commutative group with respect to addition. Further, its set of
nonzero elements $\mathcal{F}\backslash\{0\}$ is a commutative group under
multiplication. A simple yet important type of form symmetry may be defined on
a field.

\begin{definition}
Let $\mathcal{F}$ be a non-trivial field and $\{\alpha_{n}\}$ a sequence of
elements of $\mathcal{F}$ such that $\alpha_{n}\in\mathcal{F}\backslash\{0\}$
for all $n\geq-k+1.$ A (\textit{time-dependent) }\textbf{linear form symmetry}
is defined as the following special case of (\ref{tdfsk1}) with $h_{n}%
(u)=-\alpha_{n-1}u$%
\begin{equation}
\lbrack u_{0}-\alpha_{n-1}u_{1},u_{1}-\alpha_{n-2}u_{2},\ldots,u_{k-1}%
-\alpha_{n-k}u_{k}]. \label{tdfsl}%
\end{equation}
The sequence $\{\alpha_{n}\}$ of nonzero elements in $\mathcal{F}$ may be
called the \textbf{eigensequence} of the linear form symmetry. If
Eq.(\ref{dek}) has a linear form symmetry then $\{\alpha_{n}\}$ is an
eigensequence of (\ref{dek}).
\end{definition}

The use of the term \textquotedblleft eigen\textquotedblright\ which is
borrowed from the theory of linear equations is apt here for two reasons.
First, the sequence $\{\alpha_{n}\}$ characterizes the linear form symmetry
(\ref{tdfsl}) completely and secondly, we find below that linear difference
equations indeed have linear form symmetries. 

The existence of a linear form symmetry implies a type-($k,1$) order reduction
for Eq.(\ref{dek}) and a SC factorization where the cofactor equation
(\ref{tdfc1}) is determined more specifically as%
\begin{equation}
x_{n+1}=t_{n+1}+\alpha_{n}x_{n}. \label{tdfslc}%
\end{equation}

The following necessary and sufficient condition for the existence of a
time-dependent linear form symmetry is an application of Theorem \ref{tdhinv}.
We drop further mention of \textquotedblleft type-($k,1$)\textquotedblright%
\ as we do not discuss any other order reduction types in the remainder of
this paper.

\begin{corollary}
\label{tdlfs} Equation (\ref{dek}) has a time-dependent linear form symmetry
of type (\ref{tdfsl}) with an eigensequence $\{\alpha_{n}\}$ in a non-trivial
field $\mathcal{F}$ if and only if the quantity%
\begin{equation}
f_{n}(u_{0},\zeta_{1,n}(u_{0},v_{1}),\ldots,\zeta_{k,n}(u_{0},v_{1}%
,\ldots,v_{k}))-\alpha_{n}u_{0} \label{tdlfscr}%
\end{equation}
is independent of $u_{0}$ for all $n\geq0$ with the functions $\zeta_{j,n}$
for $j=1,\ldots,k$ given by
\begin{align*}
\zeta_{j,n}(u_{0},v_{1},\ldots,v_{j})  &  =\frac{\zeta_{j-1,n}(u_{0}%
,v_{1},\ldots,v_{j-1})-v_{j}}{\alpha_{n-j}}\\
&  =\frac{1}{\prod_{i=1}^{j}\alpha_{n-i}}\left(  u_{0}-\sum_{i=1}^{j}v_{i}%
{\displaystyle\prod_{p=1}^{i}}
\alpha_{n-p}\right)  .
\end{align*}

\end{corollary}

\begin{proof}
The conclusions follow immediately from Theorem \ref{tdhinv} using
$h_{n}(u)=-\alpha_{n-1}u.$ The last equality above is established from the
equality preceding it by routine calculation.
\end{proof}

\begin{remark}
If Eq.(\ref{dek}) has a linear form symmetry then by Corollary \ref{tdlfs} an
eigensequence of (\ref{dek}) can be defined equivalently as a sequence
$\{\alpha_{n}\}$ in $\mathcal{F}\backslash\{0\}$ for which the quantity in
(\ref{tdlfscr}) is independent of $u_{0}$ for every $n\geq0.$
\end{remark}

We close this section with an example of a nonlinear equation that has a linear
form symmetry. For additional results and examples, we refer to \cite{fsorbk}.

\begin{example}
Consider the following third-order nonlinear difference equation%
\begin{equation}
x_{n+1}=(-1)^{n+1}x_{n}-2x_{n-1}+g_{n}(x_{n}+x_{n-2}) \label{svag}%
\end{equation}
where $g_{n}:\mathbb{R}\rightarrow\mathbb{R}$ is a given function for each
$n\geq-2$. By Corollary \ref{tdlfs} a linear form symmetry for (\ref{svag})
exists if and only if the quantity
\begin{equation}
(-1)^{n+1}u_{0}-2\zeta_{1,n}+g_{n}(u_{0}+\zeta_{2,n})-\alpha_{n}u_{0}
\label{vagq}%
\end{equation}
is independent of $u_{0}$ for all $n$. Substituting
\[
\zeta_{1,n}=\frac{u_{0}-v_{1}}{\alpha_{n-1}},\quad\zeta_{2,n}=\frac
{\zeta_{1,n}-v_{2}}{\alpha_{n-2}}=\frac{u_{0}-v_{1}-\alpha_{n-1}v_{2}}%
{\alpha_{n-1}\alpha_{n-2}}.
\]
in (\ref{vagq}) and rearranging terms gives%
\begin{align*}
&  \left(  (-1)^{n+1}-\frac{2}{\alpha_{n-1}}-\alpha_{n}\right)  u_{0}+\frac
{2}{\alpha_{n-1}}v_{1}+\\
&  \qquad g_{n}\left(  \left(  1+\frac{1}{\alpha_{n-1}\alpha_{n-2}}\right)
u_{0}-\frac{1}{\alpha_{n-1}\alpha_{n-2}}v_{1}-\frac{1}{\alpha_{n-2}}%
v_{2}\right)
\end{align*}
which is independent of $u_{0}$ for all $n$ if the coefficients of the $u_{0}$
terms are zeros; i.e., for all $n$, the numbers $\alpha_{n}$ satisfy both of
the following equations
\begin{align}
\alpha_{n}  &  =(-1)^{n+1}-\frac{2}{\alpha_{n-1}}\label{vagr1s}\\
\alpha_{n-1}  &  =-\frac{1}{\alpha_{n-2}}. \label{vagr2s}%
\end{align}
Every solution of Eq.(\ref{vagr2s}) is a sequence of period 2%
\begin{equation}
\left\{  q,-\frac{1}{q},q,-\frac{1}{q},\ldots\right\}  \label{1oq}%
\end{equation}
where $q=\alpha_{-2}\in\mathbb{R}.$ Now (\ref{vagr2s}) yields $\alpha
_{-1}=-1/q$, which we substitute as an initial value in Eq.(\ref{vagr1s}) to
get%
\[
\alpha_{0}=-1+2q.
\]
Now to make the period-two sequence in (\ref{1oq}) also a solution of
(\ref{vagr1s}), we require the above value of $\alpha_{0}$ to be equal to $q;$
thus%
\[
\alpha_{0}=q\Rightarrow2q-1=q\Rightarrow q=1.
\]
We check that if $\alpha_{0}=q=1$ in (\ref{vagr1s}) then%
\[
\alpha_{1}=1-\frac{2}{q}=-1=-\frac{1}{q},\ \alpha_{2}=-1-\frac{2}%
{-1}=1=q,\ \text{etc}%
\]
so that both of the equations (\ref{vagr1s}) and (\ref{vagr2s}) generate the
same sequence $\{\alpha_{n}\}$ where $\alpha_{n}=(-1)^{n}$ for $n\geq-2.$ It
follows that $\{(-1)^{n}\}$ is an eigensequence for (\ref{svag}).
\end{example}

\section{Factorization of linear equations}

We expect that linear difference equations are among difference equations that
have the linear form symmetry and this is indeed the case. The following
application of Corollary \ref{tdlfs} and Theorem \ref{tdhinv} gives the
semiconjugate factorization for non-autonomous and non-homogeneous linear
difference equations.

\begin{corollary}
\label{glinsc}(The general linear equation) Let $\{a_{i,n}\}$, $i=1,\ldots,k$
and $\{b_{n}\}$ be given sequences in a non-trivial field $\mathcal{F}$ such
that $a_{k,n}\not =0$ for all $n\geq0.$ The non-homogeneous linear equation of
order $k+1$%
\begin{equation}
x_{n+1}=a_{0,n}x_{n}+a_{1,n}x_{n-1}+\cdots+a_{k,n}x_{n-k}+b_{n} \label{genlin}%
\end{equation}

has a linear form symmetry with eigensequence $\{\alpha_{n}\}$ for every
solution $\{\alpha_{n}\}$ in $\mathcal{F}$ of the following Riccati equation
of order $k$%
\begin{equation}
\alpha_{n}=a_{0,n}+\frac{a_{1,n}}{\alpha_{n-1}}+\frac{a_{2,n}}{\alpha
_{n-1}\alpha_{n-2}}+\cdots+\frac{a_{k,n}}{\alpha_{n-1}\cdots\alpha_{n-k}}
\label{ricn}%
\end{equation}

The corresponding SC factorization of (\ref{genlin}) is%
\begin{align}
t_{n+1}  &  =b_{n}-\sum_{i=1}^{k}\sum_{j=i}^{k}\frac{a_{j,n}}{\alpha
_{n-i}\cdots\alpha_{n-j}}t_{n-i+1}\label{genlinf}\\
x_{n+1}  &  =\alpha_{n}x_{n}+t_{n+1} \label{genlincf}%
\end{align}

\end{corollary}

\begin{proof}
By Corollary \ref{tdlfs} it is only necessary to determine a sequence
$\{\alpha_{n}\}$ of nonzero elements of $\mathcal{F}$ such that for each $n$
the quantity (\ref{tdlfscr}) is independent of $u_{0}$ for the following
function%
\[
f_{n}(u_{0},\ldots,u_{k})=a_{1,n}u_{0}+a_{2,n}u_{1}+\cdots+a_{k,n}u_{k}%
+b_{n}.
\]

For arbitrary $u_{0},v_{1},\ldots,v_{k}\in\mathcal{F}$ and $j=0,1,\ldots,k$
define $\zeta_{j,n}(u_{0},v_{1},\ldots,v_{j})$ as in Corollary \ref{tdlfs}.
Then the expression (\ref{tdlfscr}) is%
\begin{gather}
-\alpha_{n}u_{0}+b_{n}+a_{1,n}u_{0}+a_{2,n}\zeta_{1,n}(u_{0},v_{1}%
)+\cdots+a_{k,n}\zeta_{k,n}(u_{0},v_{1},\ldots,v_{k})=\nonumber\\
b_{n}+\left[  \sum_{j=1}^{k}\frac{a_{j,n}}{\prod_{i=1}^{j}\alpha_{n-i}}%
-\alpha_{n}\right]  u_{0}-\sum_{j=1}^{k}a_{j,n}\sum_{i=1}^{j}\frac{v_{i}%
}{\prod_{p=i}^{j}\alpha_{n-p}}\nonumber
\end{gather}

The above quantity is independent of $u_{0}$ if and only if the coefficient of
$u_{0}$ is zero for all $n;$ i.e., if $\{\alpha_{n}\}$ is a solution of the
Riccati difference equation%
\[
\alpha_{n}=\sum_{j=1}^{k}\frac{a_{j,n}}{\prod_{i=1}^{j}\alpha_{n-i}}%
\]
which is Eq.(\ref{ricn}). It follows that Eq.(\ref{genlin}) has a linear form
symmetry of type (\ref{tdfsl}) with eigensequence $\{\alpha_{n}\}$ for each
solution $\{\alpha_{n}\}$ of the Riccati equation. For the corresponding SC
factorization of (\ref{genlin}), the cofactor equation is simply
(\ref{tdfslc}) while the factor equation is obtained using the above
calculations and Eq.(\ref{tdinvf}) of Theorem \ref{tdhinv} as follows%
\begin{align*}
t_{n+1}  &  =b_{n}-\sum_{j=1}^{k}a_{j,n}\sum_{i=1}^{j}\frac{t_{n-i+1}}%
{\prod_{p=i}^{j}\alpha_{n-p}}\\
&  =b_{n}-\sum_{i=1}^{k}\sum_{j=i}^{k}\frac{a_{j,n}}{\alpha_{n-i}\cdots
\alpha_{n-j}}t_{n-i+1}.
\end{align*}
This completes the proof.
\end{proof}

Corollary \ref{glinsc} states that \textit{any} solution of the Riccati
equation (\ref{ricn}) gives a form symmetry and a SC factorization of
(\ref{genlin}) as specified above. The next example illustrates Corollary
\ref{glinsc}.

\begin{example}
Consider the second-order difference equation%
\begin{equation}
x_{n+1}=(-1)^{n+1}x_{n}+x_{n-1}+b_{n} \label{linexna}%
\end{equation}
where $b_{n},x_{0},x_{-1}$ are in a field $\mathcal{F}$ which we may take to
be any one of the familiar fields $\mathbb{Q}$, $\mathbb{R}$ or $\mathbb{C}.$
The associated Riccati equation of (\ref{linexna}) is%
\begin{equation}
\alpha_{n}=(-1)^{n+1}+\frac{1}{\alpha_{n-1}}. \label{linexnar}%
\end{equation}
Straightforward calculation shows that if $\alpha_{0}\not =0,-1$ then%
\begin{align*}
\alpha_{1}  &  =\frac{\alpha_{0}+1}{\alpha_{0}},\ \alpha_{2}=-\frac{1}%
{\alpha_{0}+1},\ \alpha_{3}=-\alpha_{0},\\
\alpha_{4}  &  =-\frac{\alpha_{0}+1}{\alpha_{0}},\ \alpha_{5}=\frac{1}%
{\alpha_{0}+1},\ \alpha_{6}=\alpha_{0}.
\end{align*}
It follows that all solutions of the Riccati equation (\ref{linexnar}) with
initial value outside the singularity set $\{0,-1\}$ are eigensequences in
$\mathcal{F}$ of period 6:%
\[
\left\{  \alpha_{0},\frac{\alpha_{0}+1}{\alpha_{0}},-\frac{1}{\alpha_{0}%
+1},-\alpha_{0},-\frac{\alpha_{0}+1}{\alpha_{0}},\frac{1}{\alpha_{0}+1}%
,\alpha_{0},\ldots\right\}
\]

The SC\ factorization of the linear equation (\ref{linexna}) is now obtained
by Corollary \ref{glinsc} as%
\begin{align*}
t_{n+1}  &  =-\frac{1}{\alpha_{n-1}}t_{n}+b_{n},\\
x_{n+1}  &  =\alpha_{n}x_{n}+t_{n+1}.
\end{align*}

\end{example}

The next result is concerned with the case of constant coefficients. The
straightforward proof is omitted.

\begin{corollary}
Let $\{b_{n}\}$ be a given sequence in a non-trivial field $\mathcal{F}$ and
let $\{a_{i}\}$, $i=1,\ldots,k$ be constants in $\mathcal{F}$ such that
$a_{k}\not =0.$

(a) The non-homogeneous linear equation of order $k+1$%
\begin{equation}
x_{n+1}=a_{0}x_{n}+a_{1}x_{n-1}+\cdots+a_{k}x_{n-k}+b_{n} \label{cclin}%
\end{equation}
has a linear form symmetry with eigensequence $\{\alpha_{n}\}$ for every
solution $\{\alpha_{n}\}$ in $\mathcal{F}$ of the following autonomous Riccati
equation of order $k$%
\begin{equation}
\alpha_{n}=a_{0}+\frac{a_{1}}{\alpha_{n-1}}+\frac{a_{2}}{\alpha_{n-1}%
\alpha_{n-2}}+\cdots+\frac{a_{k}}{\alpha_{n-1}\cdots\alpha_{n-k}}.
\label{ccric}%
\end{equation}
(b) Every fixed point of (\ref{ccric}) in $\mathcal{F}$ is a nonzero root of
the characteristic polynomial of (\ref{cclin}), i.e.,%
\begin{equation}
\lambda^{k+1}-a_{0}\lambda^{k}-a_{1}\lambda^{k-1}-\cdots-a_{k-1}\lambda-a_{k}
\label{cp}%
\end{equation}
and thus, an eigenvalue of the homogeneous part of (\ref{cclin}) in
$\mathcal{F}$. As constant solutions of (\ref{ccric}) such eigenvalues are
constant eigensequences of (\ref{cclin}).
\end{corollary}

\begin{example}
Consider the autonomous second-order linear difference equation%
\begin{equation}
x_{n+1}=x_{n}+x_{n-1}\label{linf}%
\end{equation}
Eq.(\ref{linf}) has two real eigenvalues%
\[
\alpha_{\pm}=\frac{1\pm\sqrt{5}}{2}%
\]
as roots of the characteristic polynomial $\lambda^{2}-\lambda-1$ or
equivalently, as fixed points of the Riccati equation%
\begin{equation}
\alpha_{n}=1+\frac{1}{\alpha_{n-1}}\label{ricf}%
\end{equation}
in the field $\mathcal{F=}\mathbb{R}.$ Thus each of $\alpha_{+}$ and
$\alpha_{-}$ is a constant eigensequence of (\ref{linf}) in $\mathbb{R}$ and
the following SC factorization is obtained in $\mathbb{R}$:%
\begin{align*}
t_{n+1} &  =-\frac{1}{\alpha_{+}}t_{n}=\alpha_{-}t_{n},\\
x_{n+1} &  =\alpha_{+}x_{n}+t_{n+1}.
\end{align*}
Note that the SC factorization above has constant coefficients also. We note
further that since $\alpha_{\pm}$ are irrational the above SC factorization is
not valid if $\mathcal{F=}\mathbb{Q}$ the field of rational numbers. In fact,
since the characteristic polynomial has no rational roots, it follows that
there are no constant eigensequences for (\ref{linf}) in $\mathbb{Q}$.
However, Riccati equation (\ref{ricf}) is a rational equation and thus with a
rational initial value $\alpha_{0}$ the corresponding solution of (\ref{ricf})
is a solution (non-constant) in $\mathbb{Q}.$ For instance, if $\alpha_{0}=1$
then the corresponding solution of (\ref{ricf}) is $\alpha_{n}=\varphi
_{n+1}/\varphi_{n}$ where $\{\varphi_{n}\}$ is the Fibonacci sequence
1,1,2,3,5,8,\ldots\ This rational eigensequence yields the following SC
factorization of (\ref{linf}) that is valid in $\mathbb{Q}$:%
\begin{align*}
t_{n+1} &  =-\frac{\varphi_{n-1}}{\varphi_{n}}t_{n},\\
x_{n+1} &  =\frac{\varphi_{n+1}}{\varphi_{n}}x_{n}+t_{n+1}.
\end{align*}
We note that $\lim_{n\rightarrow\infty}\varphi_{n+1}/\varphi_{n}=\alpha_{+}$
in the above factorization; in this way the factorization over rationals is
related to the earlier factorization over the reals. In a similar fashion, the
equation%
\begin{equation}
x_{n+1}=x_{n}-x_{n-1}\label{linf0}%
\end{equation}
has two complex eigenvalues%
\[
\alpha_{\pm}=\frac{1\pm i\sqrt{3}}{2}%
\]
that are roots of $\lambda^{2}-\lambda+1.$ Thus, (\ref{linf0}) has no constant
eigensequences in $\mathbb{R}$ but it does have non-constant real
eigensequences since the Riccati equation%
\[
\alpha_{n}=1-\frac{1}{\alpha_{n-1}}%
\]
with the initial value $\alpha_{0}=2$ has a solution%
\[
\left\{  2,\frac{1}{2},-1,2,\frac{1}{2},-1,\ldots\right\}
\]
of period three in $\mathbb{R}$ with a corresponding real SC\ factorization%
\[
t_{n+1}=-\frac{1}{\alpha_{n-1}}t_{n},\quad x_{n+1}=\alpha_{n}x_{n}+t_{n+1}.
\]
In contrast to the factorization of Eq.(\ref{linf}) there is no simple relationship
between the factorization of (\ref{linf0}) over the real numbers and its factorization 
with constant eigensequences over the complex numbers.
\end{example}

\begin{remark}
Is it possible that a linear difference equation has no eigensequences,
constant or otherwise in a given field $\mathcal{F}$ because the associated
Riccati equation has no solutions at all in $\mathcal{F}?$

We know the answer to this question in some cases. If we have a linear
equation (homogeneous or not) with constant coefficients in an algebraically
closed field $\mathcal{F}$ (e.g., the field $\mathbb{C}$ of complex numbers)
then $\mathcal{F}$ always contains constant eigensequences, namely, the roots
of the characteristic polynomial (\ref{cp}). On the other hand, for the finite
field $\mathbb{Z}_{3}=\{0,1,2\}$ with addition and multiplication defined
modulo 3, the linear equation (\ref{linf}) has no eigensequences. This can be
shown by testing each of the two possible nonzero initial values 1,2 in the
Riccati equation (\ref{ricf}) to verify that both lead to the singularity at 0:%
\begin{align*}
\alpha_{0}  &  =2\Rightarrow\alpha_{1}=1+\frac{1}{2}=1+2=0,\\
\alpha_{0}  &  =1\Rightarrow\alpha_{1}=1+1=2\Rightarrow\alpha_{2}=0.
\end{align*}
The answer to the question of existence of eigensequences in the general case
is not known at this time; in fact, it is not known if a linear equation with
real coefficients exists that has no real eigensequences. For
\textquotedblleft large\textquotedblright\ fields such as $\mathbb{R}$ or
$\mathbb{C}$ it seems likely that the general linear equation (\ref{genlin})
has an eigensequence in the field.
\end{remark}

The occurrence of Riccati difference equation in Corollary \ref{glinsc} may
seem less surprising if we recall some basic facts from \cite{hsarx}. In
particular, the homogeneous part of (\ref{genlin}) is a homogeneous equation
of degree one relative to the multiplicative group $\mathcal{F}\backslash
\{0\}.$ Therefore, it has an inversion form symmetry and the factor equation
of its SC factorization is none other than the Riccati equation (\ref{ricn}).
Using this fact it is possible to restate Corollary \ref{glinsc} without
explicit reference to the Riccati equation as follows.

\begin{corollary}
\label{hom}Assume that the homogeneous part of Eq.(\ref{genlin}) has a
solution $\{y_{n}\}$ in the field $\mathcal{F}$ such that $y_{n}\not =0$ for
all $n.$ Then $\{y_{n+1}/y_{n}\}$ is an eigensequence of (\ref{genlin}) whose
SC factorization is given by the pair of equations (\ref{genlinf}) and
(\ref{genlincf}).
\end{corollary}

\begin{proof}
It is given
that $\{y_{n}\}$ satisfies the homogeneous part of (\ref{genlin}), i.e.,%
\[
y_{n+1}=a_{0,n}y_{n}+a_{1,n}y_{n-1}+a_{2,n}y_{n-2}+\cdots+a_{k,n}y_{n-k}.
\]

Since $y_{n}\not =0$ for all $n,$ we may divide the above equation by $y_{n}$
to obtain%
\begin{align*}
\frac{y_{n+1}}{y_{n}}  &  =a_{0,n}+a_{1,n}\frac{y_{n-1}}{y_{n}}+a_{2,n}%
\frac{y_{n-2}}{y_{n}}+\cdots+a_{k,n}\frac{y_{n-k}}{y_{n}}\\
&  =a_{0,n}+a_{1,n}\frac{y_{n-1}}{y_{n}}+a_{2,n}\frac{y_{n-2}}{y_{n-1}}%
\frac{y_{n-1}}{y_{n}}+\cdots+a_{k,n}\frac{y_{n-k}}{y_{n-k+1}}\cdots
\frac{y_{n-1}}{y_{n}}.
\end{align*}

Now defining $\alpha_{n}=y_{n+1}/y_{n}$ for all $n$ and substituting these
terms in the last equation above yields the Riccati equation (\ref{ricn}).
Thus $\{y_{n+1}/y_{n}\}$ is an eigensequence of (\ref{genlin}) in
$\mathcal{F}\backslash\{0\}$, as claimed. The SC factorization is obtained as
in the proof of Corollary \ref{glinsc}.
\end{proof}

\begin{corollary}
In Eq.(\ref{genlin}) let $\{a_{i,n}\}$, $i=1,\ldots,k$ and $\{b_{n}\}$ be
sequences of real numbers with $a_{i,n}\geq0$ for all $i,n$ and $a_{k,n}>0$
for all $n.$ Then (\ref{genlin}) has an eigensequence $\{y_{n+1}/y_{n}\}$ and
a SC factorization in $\mathbb{R}$ given by the pair of equations
(\ref{genlinf}) and (\ref{genlincf}).
\end{corollary}

\begin{proof}
If we choose $y_{-j}=1$ for $j=0,\ldots,k$ then the corresponding solution
$\{y_{n}\}$ of the homogeneous part of (\ref{genlin}) is a sequence of
positive real numbers. Now an application of Corollary \ref{hom} completes the proof.
\end{proof}

\end{document}